\documentclass[showkeys,floatfix,aps,onecolumn,notitlepage,a4paper,superscriptaddress,nofootinbib]{revtex4-1}
\usepackage[hidelinks,breaklinks=true]{hyperref}
\usepackage[english]{babel}
\usepackage{amsopn,amsthm,dsfont,amssymb}
\usepackage{mathtools}
\usepackage{algorithm}
\usepackage{algpseudocode}
\mathtoolsset{showonlyrefs=true}
\usepackage{adjustbox}
\usepackage[caption=false]{subfig}
\newtheorem{theorem}{Theorem}
\newtheorem{corollary}{Corollary}
\theoremstyle{definition}
\newtheorem{definition}{Definition}
\newcommand{\cc}[1]{\operatorname{\mathsf{#1}}}

\begin{document}
\title{Computational tameness of classical non-causal models}
\author{\"Amin Baumeler}
\affiliation{Faculty of Informatics, Universit\`{a} della Svizzera italiana, Via G.\ Buffi 13, 6900 Lugano, Switzerland}
\affiliation{Facolt\`{a} indipendente di Gandria, Lunga scala, 6978 Gandria, Switzerland}
\author{Stefan Wolf}
\affiliation{Faculty of Informatics, Universit\`{a} della Svizzera italiana, Via G.\ Buffi 13, 6900 Lugano, Switzerland}
\affiliation{Facolt\`{a} indipendente di Gandria, Lunga scala, 6978 Gandria, Switzerland}
\begin{abstract}
	We show that the computational power of the non-causal circuit model, {\it i.e.}, the circuit model where the assumption of a {\em global causal order\/} is replaced by the assumption of {\em logical consistency}, is completely characterized by the complexity class~$\cc{UP}\cap\cc{coUP}$.
	An example of a problem in that class is factorization.
	Our result implies that classical deterministic closed timelike curves (CTCs) cannot efficiently solve problems that lie {\em outside\/} of that class.
	Thus, in stark contrast to other CTC models, these CTCs {\em cannot\/} efficiently solve~$\cc{NP-complete}$ problems, unless~$\cc{NP}=\cc{UP}\cap\cc{coUP}=\cc{coNP}$, which lets their existence in nature appear {\em less implausible}.
	This result gives a new characterization of~$\cc{UP}\cap\cc{coUP}$ in terms of fixed points.
\end{abstract}

\maketitle
\section{Motivation and results}
The {\em acyclic\/} feature of ``causality''~\cite{Pearl}, that an effect cannot be the cause of its cause, plays a central role in everyday live, physical theories, and models of computation.
A {\em cyclic\/} causal structure is --- in the classical
meaning\footnote{The noun ``paradox'' means a {\em seeming\/} contradiction as opposed to an {\em actual\/} contradiction.
	It originates from the Greek word {\em paradoxon\/} which is composed out of {\em para\/} (against) and {\em doxa\/} (opinion).
	We use the term {\em antinomy\/} for actual contradictions.}
of the following adjective --- paradoxical.
That may be a reason for why an {\em acyclic\/} notion is not only preferred but also a hidden assumption for many theories.
Objections against {\em cyclic\/} causal structures are the {\em grandfather antinomy\/} and the {\em uniqueness ambiguity\/}\footnote{The uniqueness ambiguity is also known under the name ``information paradox'' or ``information antinomy.''} (see,~{\it e.g.}, References~\cite{AaronsonWhy,Pienaar2013,Allen2014}).
The former reads: By travelling to the past and killing his or her own grandfather, one could never have been born to travel to the past to kill his or her own grandfather --- an {\em inconsistency}.
The latter is {\it ex nihilo\/} appearance of information, as illustrated in the following example.
Assume one morning you wake up to find a proof of~$\cc{P}=\cc{NP}$ on your desk.
You decide to publish it and, after publication, you travel back in time to the night before you found the proof to place the original copy on your desk, while your younger self is asleep.
Who wrote the proof?
More precisely, the uniqueness ambiguity arises when ``an uncomputed output is produced,'' in the sense that some ``theory specifies more the one final state given some initial state and evolution, but fails to give probabilities for each possibility''~\cite{Allen2014}. 
However, if the proof you find on your desk is {\em uniquely\/} determined by a process, then the proof does not appear {\it ex nihilo}, but is the {\em result\/} of that process.
The uniqueness ambiguity is often considered {\em less severe\/} than the grandfather antinomy.
But note that, according to Deutsch~\cite{Deutsch:1991jo}, solutions to problems need to emerge through evolutionary or rational processes; otherwise, the underlying theory would follow the {\em doctrine of creationism}.
By this, uniqueness ambiguities ``contradict the philosophy of science.''
Note that both problems are similar in their spirit; the grandfather antinomy, in accordance to Allen's~\cite{Allen2014} formulation of the uniqueness ambiguity, reads: The grandfather antinomy arises whenever a theory fails to specify {\em any consistent\/} final state given some initial state and evolution.
In the following, we will refer to a model as being {\em logically consistent\/} whenever both problems do not arise.

Closed timelike curves (CTCs) are loops in spacetime (see Figure~\ref{fig:ctcfuture}).
\begin{figure}
        \centering
        \subfloat[\label{fig:ctcfuture}]{
		\includegraphics{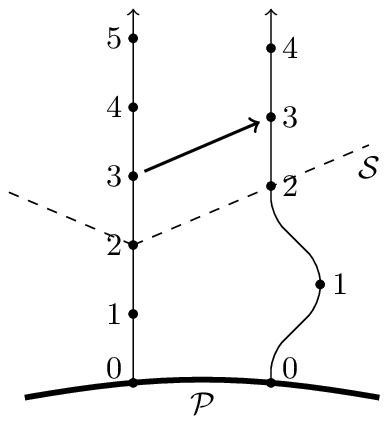}
        }
	\qquad
	\qquad
        \subfloat[\label{fig:ctcnow}]{
		\includegraphics{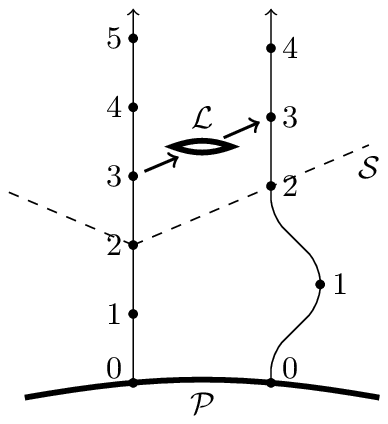}
	}
	\caption{(a) Example of a CTC created from a wormhole~\cite{Morris1988,inprep}. The lines from bottom to top represent the worldlines of two mouths of a wormhole in a space-time diagram. Proper times of the mouths are identified. The right mouth undergoes some time dilation in order to generate a CTC in the future of~$\mathcal S$. Thus, by traveling through ordinary space from point~3 on the left worldline to point~3 on the right, and by then entering the wormhole, one exits at point~3 on the left again.
		(b) An experimenter sitting in the region labeled by~$\mathcal L$ can freely manipulate the degrees of freedom of an object traveling on the CTC.
	}
	\label{fig:CTC}
\end{figure}
That is, by traveling on such a curve, one would bump into oneself on the same position in space {\em and\/} time.
Interestingly, CTCs appear as solutions to Einstein's equations of general relativity (see,~{\it e.g.}, References~\cite{Lanczos1924,Godel:1949eb,Taub1951,Newman1962,Kerr1963,Tipler1974,Griffiths2009}), yet they have been or still are believed to be unphysical; their underlying structure is {\em cyclic}.
For over twenty years, people have studied different models of CTCs and their implications.
Scientists around Novikov and Thorne~\cite{Friedman1990,Echeverria:1991ko,Friedman:1990ja,Lossev1992,Novikov1992,Mikheeva1993} analyzed CTCs in the gravitational setting and found self-consistent dynamics for all initial conditions considered.
In more detail, they studied the trajectories of objects like billiard balls that, once the initial conditions of the objects have been specified (that is, on the surface~$\mathcal P$ in Figure~\ref{fig:ctcfuture}), travel trough CTCs and bounce off themselves.
Their result is surprising: {\em Multiple}, as opposed to {\em zero\/} (``to one's naive expectation''~\cite{Echeverria:1991ko}), self-consistent trajectories to initial conditions that lead to self-collisions were found.
Self-{\em in\/}consistent trajectories are simply neglected by the means of Novikov's principle of self-consistency~\cite{Friedman:1990ja}.
While the grandfather antinomy is avoided, the uniqueness ambiguity persists.
Deutsch~\cite{Deutsch:1991jo} analyzed CTCs in the quantum information realm and showed that there, the grandfather antinomy never occurs.
Because {\em multiple\/} consistent states to some initial conditions exist, Deutsch singles out the mixture of all consistent states that maximizes the entropy as the solution; by this he mitigates the uniqueness ambiguity.
This maximum-entropy strategy, however, has a price: The evolution becomes non-linear.
So, the self-consistent state might be a mixed state, by which one is forced to consider mixed states as {\em ontic}.
This means that the states of the systems traveling on Deutsch CTCs describe ``reality'' as opposed to the {\em knowledge of an observer about a system\/}~\cite{Wallman2012}.
Bennett, Leung, Smith and Smolin~\cite{Bennett2009} criticized the results on the computational power of Deutsch CTCs by pointing at a ``linearity trap:'' If one uses a mixture of problems as input to some CTC, one is not given the mixture of the solutions.
By this, in similar spirit to our work and to Reference~\cite{inprep}, they define a (possibly weaker) CTC model where input-output pairs are correlated correctly.
Note that the present result is not akin to the ``linearity trap,'' as our underlying models are linear.
Pegg~\cite{Pegg:2001wa} and others~\cite{Bennett,Svetlichny:2009ve,Lloyd:2011ir,Svetlichny:2011gq,Ralph:2012cd} designed a different model of CTCs, in which states are sent with the help of quantum teleportation to the past (via postselection).
That model, however, also leads to a non-linear evolution.
Recently, Oreshkov, Costa, and Brukner~\cite{Oreshkov2012} came up with a framework for quantum correlations without global causal order.
There, the main assumptions are {\em linearity\/} and {\em local validity\/} of quantum theory.
Interestingly, the framework describes correlations that cannot be simulated with a {\em global\/} causal order~\cite{Oreshkov2012,Baumeler3parties,simplestcausalinequalities,inprep}, and allows for advantages in query~\mbox{\cite{Chiribella2012,Colnaghi2012,Chiribella2013,Araujo2014,Procopio2015}}, as well as communication complexity~\cite{Feix2015,Guerin2016}.
The classical special case~\cite{Baumeler2016} of that framework was shown to allow for {\em classical deterministic CTCs\/}~\cite{inprep} where both problems (grandfather antinomy and uniqueness ambiguity) {\em never\/} arise\footnote{It is not the case that the problems are {\em concealed\/} due to {\em lack of knowledge\/}~\cite{Wallman2012}, but they do not even arise on the {\em ontic\/} level. Note that Wallman and Bartlett~\cite{Wallman2012} furthermore studied Deutsch CTCs where the underlying systems are taken from Spekkens'~\cite{Spekkens2007} ``toy theory'' of quantum theory.}; therefore, we refer to these CTCs as {\em logically consistent CTCs}.
The main conceptual difference to the works by Novikov and Thorne are that in setups with logically consistent CTCs, experimenters are {\em free\/} to manipulate the classical systems that travel on closed time-like curves, as opposed to be restricted in only choosing the {\em initial conditions\/} (see Figure~\ref{fig:ctcnow}).
One can also define a {\em non-causal circuit model\/} of computation~\cite{BaumelerNC} based on the assumption that both problems are avoided.
Here, we characterize the computational power of that circuit model which yields, as we are going to show, an upper bound on the computational power of {\em classical deterministic CTCs}.

Even though we do not know whether CTCs exist in nature or not, we can study their consequences.
As Aaronson~\cite{Aaronsonsciam} put it, one could assume that nature {\em cannot\/} efficiently solve certain tasks ({\it e.g.},~$\cc{NP-hard}$ problems), in the same spirit as nature cannot signal faster than at the speed of light, and conclude that certain theories are {\em un\/}physical.
The same idea is used in reconstructions of quantum theory where the standard, unintuitive axioms are replaced by ``more natural'' ones (see, {\it e.g.}, Reference~\cite{foils} for a collection of such reconstructions).
As it turns out~\cite{Aaronson:2009dy}, the class~$\cc{P_\text{CTC}}$ of all problems solvable in polynomial time by classical Deutsch CTCs is equal to its quantum analog~$\cc{BQP_\text{CTC}}$, and furthermore, equal to~$\cc{PSPACE}$.\footnote{Some intuition behind this result is that Deutsch CTCs make time {\em reusable\/} just as space is, and thus a polynomial amount of space equals a polynomial amount of {\em reusable\/} time~\cite{reusable}. That Deutsch CTCs can solve difficult computational problems efficiently was also pointed out by others (see, {\it e.g.}, References~\cite{Brun2003,Bacon2004}).}
Most recently, Aaronson, Bavarian, and Gueltrini~\cite{haltingp} showed that the Deutsch model can even solve the halting problem.
The model of CTCs where the loops are generated through quantum teleportation to the past can efficiently solve all problems in the class~\mbox{$\cc{PostBQP}=\cc{PP}$~\cite{Aaronson2005PostBQPPP,Lloyd:2011ir,Lloyd2011exp}}.
The classical analogue thereof can efficiently solve problems in~$\cc{PostBPP}=\cc{BPP_\text{path}}$~\cite{Han1997,Lloyd:2011ir}.
The inclusion relations between these classes are~$\cc{NP}\subseteq\cc{PostBPP}\subseteq\cc{PostBQP}\subseteq\cc{P_\text{CTC}}\subseteq\cc{EXP}$, where {\em strict\/} inclusions are conjectured.
Our contribution is to show
\begin{align}
	\cc{P_\text{LCCTC}}\subseteq\cc{P_\text{NCCirc}}=\cc{UP}\cap\cc{coUP}
	\,,
\end{align}
{\it i.e.}, 
that the class~$\cc{P_\text{NCCirc}}$ (NCCirc standing for ``non-causal circuit'') of decision problems solvable in polynomial time with the {\em non-causal circuit model\/} is equal to~$\cc{UP}\cap\cc{coUP}$,
and furthermore that~$\cc{P_\text{LCCTC}}$, which represents the power of classical computation equipped with logically consistent CTCs~\cite{inprep}, is upper bounded by this class.
The class~$\cc{UP}\cap\cc{coUP}$ contains all decision problems where for each possible answer (``yes'' or ``no'') a {\em unique\/} witness exists. 
Examples of such problems are integer factorization~\cite{Fellows1992} and parity games~\cite{Jurdzinski1998parity}, casted as decision problems.
Thus, this class is of great importance to the field of cryptography.
Moreover, it was shown~\cite{Homan2003} that worst-case one-way permutation exist if and only if~$\cc{P}\not=\cc{UP}\cap\cc{coUP}$.
Figure~\ref{fig:cc} depicts the inclusion relations among the mentioned complexity classes: ~$\cc{P}\subseteq\cc{P_\text{LCCTC}}\subseteq\cc{P_\text{NCCirc}}\subseteq\cc{NP}\subseteq\cc{PostBPP}\subseteq\cc{PostBQP}\subseteq\cc{P_\text{CTC}}$ (see Figure~\ref{fig:cc}).
\begin{figure}
	\centering
	\includegraphics{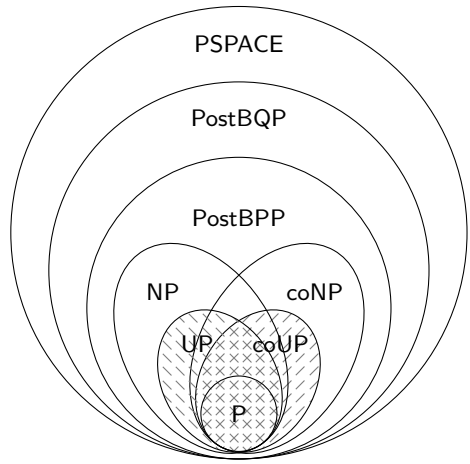}
	\caption{Inclusion relations among complexity classes under consideration. Deutsch CTCs (classical and quantum) can efficiently solve every problem in~$\cc{PSPACE}$, quantum CTCs with postselected teleportation every problem in~$\cc{PostBQP}$, classical CTCs with postselected teleportation every problem in~$\cc{PostBPP}$, and the classical non-causal circuit model every problem in~$\cc{UP}\cap\cc{coUP}=\cc{P_\text{NCCirc}}$ (region marked with crosses). Logically consistent CTCs {\em cannot\/} efficiently solve problems outside of the latter class.}
	\label{fig:cc}
\end{figure}
The logically consistent CTCs~\cite{inprep} are the weakest of all known CTCs in terms of computation, and are {\em unable\/} to efficiently solve~$\cc{NP-complete}$ problems (unless~$\cc{NP}=\cc{UP}\cap\cc{coUP}$, which implies~$\cc{NP}=\cc{coNP}$, by which the polynomial hierarchy would collapse to the first level~\cite{Stockmeyer1976}, which is highly doubted).
We also show the analog statement for {\em search\/} problems:
\begin{align}
	\cc{FP_\text{NCCirc}}=\cc{F}(\cc{UP}\cap\cc{coUP})=\cc{TFUP}
	\,,
\end{align}
where~$\cc{TFUP}$ is the class of all search problems with {\em unique\/} solutions.
Furthermore, these results give an interpretation of the classes~$\cc{UP}\cap\cc{coUP}$ and~$\cc{TFUP}$ in terms of {\em fixed points}: Every instance of such a problem can be solved by finding the {\em unique\/} deterministic fixed point of a transformation computable in polynomial time.

This work is organized as follows.
First, we describe the computational model, and after that, we define some complexity classes and present our results.
Then, we present an example on how to factorize integers by using that model, give conclusions, and state some open problems.

\section{Model of computation}
Classical deterministic CTCs that are free of the grandfather antinomy and the uniqueness ambiguity were studied in Reference~\cite{inprep}.
There it was shown that such CTCs are logically possible even in the case where~$N$ parties sitting in localized regions can {\em freely\/} interact with the systems travelling on the CTCs.
We ask the reader to consult Figure~\ref{fig:ourctcs} with the following description of the CTC model.
\begin{figure}
	\centering
	\includegraphics{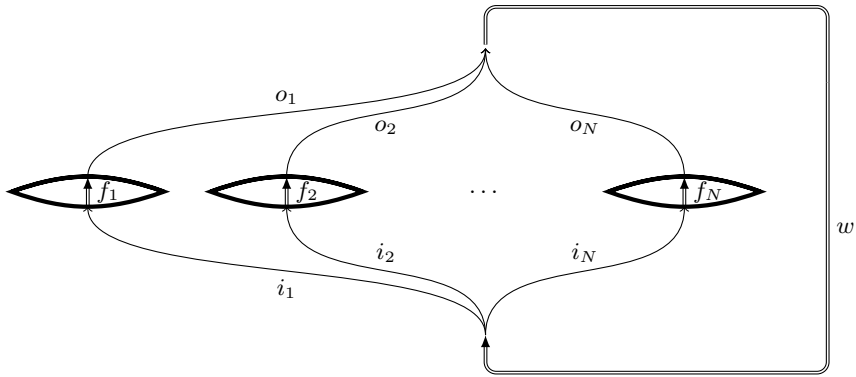}
	\caption{Each of the~$N$ local regions has a past and a future boundary. The state space on the past boundary for region~$j$ is~$\mathcal I_j$, and state space on the future boundary for the same region is~$\mathcal O_j$. The party within region~$j$ implements some function~$f_j:\mathcal I_j\rightarrow \mathcal O_j$. The output state of all regions undergoes some CTC transformation~$w$ (``process function'').}
	\label{fig:ourctcs}
\end{figure}
Let~$\mathcal I_j$ ($\mathcal O_j$) be the state spaces of the past (future) boundary of the region of some party~$1\leq j\leq N$.
Party~$j$ implements some function~$f_j$ of her {\em choice\/} from the set~$\mathcal D_j:=\{f_j:\mathcal I_j \rightarrow O_j\}$ of all functions from~$\mathcal I_j$ to~$\mathcal O_j$.
Thus, all~$N$ parties implement some function~$f:\mathcal I\rightarrow\mathcal O$.
Here, and in what follows, we drop indices in order to refer to {\em collection\/} of objects in all~$N$ regions, {\it e.g.},~$\mathcal I=\mathcal I_1\times\mathcal I_2\times\dots\times\mathcal I_N$.
In that setting, it was shown that a CTC free of the problems discussed is represented by some ``process function''~$w:\mathcal O\rightarrow\mathcal I$ where\footnote{We use~$\exists!$ to refer to the {\em uniqueness\/} quantifier.}
\begin{align}
	\forall f\in\mathcal D,\exists! i\in\mathcal I: w\circ f(i)=i
	\,.
	\label{eq:lc}
\end{align}
In words, the ``process function'' composed with {\em any choice\/} of the local operations results in a function that has a {\em unique\/} deterministic fixed point.
This is easily interpreted: If for some choice of local operations there would be {\em no\/} fixed point, then the grandfather antinomy is reproduced, if there are {\em more than one\/} fixed points, then the uniqueness ambiguity is reproduced.
For three or more local regions, such antinomy-free CTCs become possible~\cite{inprep} (in the sense that there exist local operations and ``process functions'' where a region necessarily is {\em in the past and in the future\/} of every other region.

The non-causal circuit model~\cite{BaumelerNC}, then again, is formulated in terms of {\em gates\/} as opposed to {\em parties\/} and ``process functions.''
A {\em circuit\/} is a collection of gates that are connected in an acyclic fashion, and where the input and output wires are numbered from~$1$ on upwards in integer steps (see Figure~\ref{fig:circuit}).
\begin{figure}
	\centering
	\subfloat[\label{fig:circuit}]{
		\includegraphics{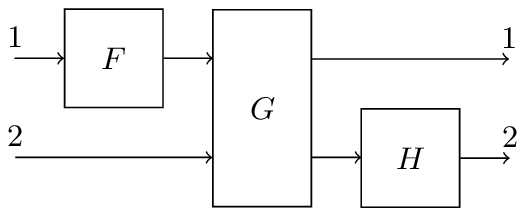}
	}
	\qquad
	\qquad
	\subfloat[\label{fig:circuitclosed}]{
		\includegraphics{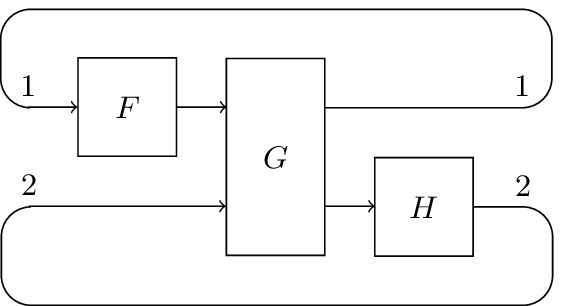}
	}
	\caption{(a) Example of a circuit where the input and output wires are labeled by~$1,2$. (b) Closed circuit constructed from~(a).}
\end{figure}
Without loss of generality, and if not otherwise stated, we assume that every wire carries a bit.
A {\em closed circuit\/} is a circuit without input and without output wires.
A circuit~$\mathcal{C}$ with the same number of input and output wires is transformed to a closed circuit~$\mathcal{C}'$ by connecting all input and output wires with the same label (see Figure~\ref{fig:circuitclosed}).
The introduced connections can thought as ``back in time'' --- in the same spirit as the ``back in time'' connections in CTC models.
Let~$c:\mathcal X\rightarrow\mathcal X$ be the function that is induced by the circuit~$\mathcal{C}$ (the set~$\mathcal Y$ is~$\{0,1\}^n$ where~$n$ is the number of input bits to~$C$).
We call a closed circuit~$\mathcal{C}'$ {\em logically consistent\/} if and only if~$\mathcal{C}$ has a unique deterministic fixed point,~{\it i.e.},
\begin{align}
	\exists! y\in\mathcal Y:c(y)=y
	\,.
	\label{eq:nccirc}
\end{align}
The difference between this model and the CTCs discussed above is that here, the gates are {\em fixed\/} whereas for the CTCs, every party can {\em arbitrarily\/} choose her local operation.
Thus, we omit the all quantifier in the logical-consistency condition for circuits (compare the above Equation with Equation~\eqref{eq:lc}).
Logically consistent closed circuits can be used to find unique fixed points, which is exploited in what follows.

\section{Complexity classes}
A {\em decision problem\/}~$\Pi$ is often casted as the membership problem of a language~$L\subseteq\Sigma^*$ with alphabet~$\Sigma$.
For simplicity, and without loss of generality, we choose~$\Sigma=\{0,1\}$.
An instance of~$\Pi$ is a string~$x\in\Sigma^*$, and the question is: Is~$x$ a word of~$L$, {\it i.e.}, does~$x\in L$ hold?
An algorithm that solves a decision problem outputs either~``yes'' or~``no.''

{\em Search problems}, then again, are mostly defined via binary relations.
A problem~$\Pi$ is associated with a binary relation~$R\subseteq\Sigma^*\times\Sigma^*$.
An instance of~$\Pi$ is some~$x\in\Sigma^*$, and the question is: {\em What\/} (if there exists one) is~$y\in\Sigma^*$ such that~$(x,y)\in R$?
An algorithm that solves a search problem outputs~$y$ if there exists a~$y$ satisfying~$(x,y)\in R$, and returns~``no'' otherwise.

We use~$|x|$ to denote the length of some string~$x\in\Sigma^*$.
A binary relation~$R$ is called {\em polynomially decidable\/} if there exists a deterministic Turing machine deciding the language~$\{(x,y)\in R\}$ in polynomial time, and~$R$ is called {\em polynomially balanced\/} if there exists some polynomial~$q$ such that~$(x,y)\in R$ implies~$|y|\leq q(|x|)$.

In the following definitions of complexity classes, we require that for every problem~$\Pi$ and given a string~\mbox{$x\in\Sigma^*$}, we can check in polynomial time whether~$x$ is an instance of~$\Pi$ or not.
If~$x$ is {\em not\/} an instance of~$\Pi$, then we abort.
We refer the reader to References~\cite{Papadimitriou1995,Arora2009} for common concepts in complexity theory.
\begin{definition}[Deterministic NCCirc algorithm]
	A {\em deterministic NCCirc algorithm\/}~$\mathcal{A}$ is a polynomial time deterministic algorithm that takes as input some bit string~$x\in\{0,1\}^*$ and outputs a Boolean circuit~$\mathcal{C}_x$ over AND, OR, and NOT, such that for every~$x$ the closed circuit~$\mathcal{C}'_x$ is logically consistent, {\it i.e.},
	\begin{align}
		\forall x\in\{0,1\}^*,\exists!y: c_x(y)=y
		\,.
	\end{align}
	If the fixed point~$y$ has the form~$y=1z$ for some~$z$, then we say~$\mathcal{A}$ {\em accepts\/}~$x$, otherwise,~$\mathcal{A}$ {\em rejects\/}~$x$.
	The algorithm~$\mathcal{A}$ {\em decides a language\/}~$L\subseteq\{0,1\}^*$ if~$\mathcal{A}$ accepts every~$x\in L$ and rejects every~$x\not\in L$.
	Furthermore, the algorithm~$\mathcal{A}$ {\em decides a binary relation\/}~$R\subseteq\{0,1\}^*\times\{0,1\}^*$ if for every~$x\in\{0,1\}^*$ the pair~$(x,y)$, with~$c_x(y)=y$, is in~$R$.
\end{definition}

Based on the above definition, we define the complexity classes~$\cc{P_\text{NCCirc}}$ and~$\cc{FP_\text{NCCirc}}$.
\begin{definition}[$\cc{P_\text{NCCirc}}$ and~$\cc{FP_\text{NCCirc}}$]
	The class~$\cc{P_\text{NCCirc}}$ contains all languages decidable by some deterministic NCCirc algorithm.
	The class~$\cc{FP_\text{NCCirc}}$ contains all binary relations decidable by some deterministic NCCirc algorithm.
\end{definition}
We will relate~$\cc{P_\text{NCCirc}}$ to the following complexity class.
\begin{definition}[$\cc{UP}$]
	The class~$\cc{UP}$ (Unambiguous Polynomial-time) contains all languages~$L$ for which a polynomial-time verifier~$V:\{0,1\}^*\times\{0,1\}^*\rightarrow\{0,1\}$ exists such that for every~$x$, if~$x\in L$ then~$\exists!y:V(x,y)=1$, and if~$x\not\in L$ then~$\forall y:V(x,y)=0$.
\end{definition}
The complexity class~$\cc{UP}$ was first defined by Valiant~\cite{Valiant1976}.
The only difference between the classes~$\cc{NP}$ and~$\cc{UP}$ is that in the former, {\em multiple\/} witnesses are allowed.
The class~$\cc{coUP}$ contains all languages~$L$ where the complement of~$L$ is in~$\cc{UP}$.

We are now ready to state our first theorem.
\begin{theorem}
	$\cc{P_\text{NCCirc}}=\cc{UP}\cap\cc{coUP}$.
\end{theorem}
\begin{proof}
	We start by showing~$\cc{UP}\cap\cc{coUP}\subseteq\cc{P_\text{NCCirc}}$.
	Assume a language~$L$ is in~$\cc{UP}\cap\cc{coUP}$.
	Thus, there exist two polynomial-time verifiers~$V_\text{yes}$ and~$V_\text{no}$ such that for every~$x$, if~$x\in L$, then
	\begin{align}
		\exists!w:V_\text{yes}(x,w)=1 \wedge \forall w':V_\text{no}(x,w')=0
		\,,
	\end{align}
	and otherwise
	\begin{align}
		\forall w:V_\text{yes}(x,w)=0 \wedge \exists!w':V_\text{no}(x,w')=1
		\,.
	\end{align}
	The following deterministic NCCirc algorithm~$\mathcal{A}$ decides the language~$L$.
	Upon receiving~$x\in\{0,1\}^*$,~$\mathcal{A}$ generates the circuit~$\mathcal{C}_x$ as shown in Figure~\ref{fig:CX}.
	\begin{figure}
		\centering
		\includegraphics{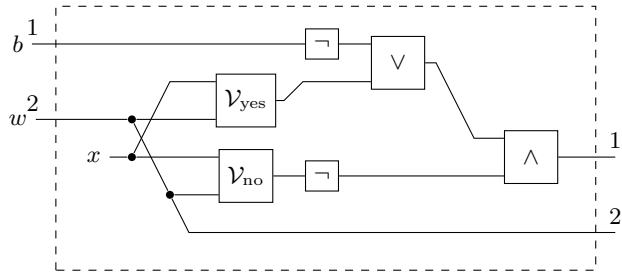}
		\caption{Circuit~$\mathcal{C}_x$ used to reduce a problem from~$\cc{UP}\cap\cc{coUP}$ to~$\cc{P_\text{NCCirc}}$. The wire that carries~$w$ consists of~$q(|x|)$ bits.}
		\label{fig:CX}
	\end{figure}
	The subcircuits~$\mathcal{V}_\text{yes},\mathcal{V}_\text{no}$ implement the verifiers~$V_\text{yes},V_\text{no}$, and can be constructed in polynomial time, because~$L$ is assumed to be in~$\cc{UP}\cap\cc{coUP}$.
	The circuit acts in the following way:
	\begin{align}
		c_x&:\{0,1\}\times\{0,1\}^{q(|x|)} \rightarrow \{0,1\}\times\{0,1\}^{q(|x|)}\,,\\
		&:(b,w)\mapsto
		\begin{cases}
			(0,w)&V_\text{no}(x,w)=1\,,\\
			(1,w)&V_\text{yes}(x,w)=1\,,\\
			(b\oplus 1,w)&\text{otherwise,}
		\end{cases}
	\end{align}
	where~$q$ is a polynomial.
	The function~$c_x$ has a {\em unique\/} fixed point.
	If~$x\in L$, then there exists a unique~$w$ with~$V_\text{yes}(x,w)=1$, and~$c_x(1w)=1w$.
	Otherwise, there exists a unique~$w$ with~$V_\text{no}(x,w)=1$, and~$c_x(0w)=0w$.

	The converse~($\cc{P_\text{NCCirc}}\subseteq\cc{UP}\cap\cc{coUP}$) holds for the following reason.
	First, assume~$L$ is in~$\cc{P_\text{NCCirc}}$.
	This means that for every~$x$ we have some logically consistent circuit~$\mathcal{C}'_x$.
	We design both verifiers~$V_\text{yes}$ and~$V_\text{no}$ to act as
	\begin{align}
		V_\text{yes}:&(x,z)\mapsto c_x(z)=z\wedge z=1w\,,\\
		V_\text{no}:&(x,z)\mapsto c_x(z)=z\wedge z=0w\,.
	\end{align}
	That is, both verifiers check whether~$z$ is a fixed point of~$\mathcal{C}_x$, and additionally check for the first bit.
\end{proof}

A corollary of this Theorem is that logically consistent CTCs cannot efficiently solve problems {\em outside\/} of the class~$\cc{UP}\cap\cc{coUP}$.
To state this corollary, we first define the complexity class~$\cc{P_\text{LCCTC}}$ of problems efficiently solvable by such CTCs.
\begin{definition}[Classical deterministic CTC algorithm and $\cc{P_\text{LCCTC}}$]
	A {\em classical deterministic CTC algorithm\/}~$\mathcal B$ is a polynomial time deterministic algorithm that takes as input some bit string~$x\in\{0,1\}^*$ and outputs~$N$ (the number of parties in the CTC setup), a list of non-negative integers $(m_1,n_1,m_2,n_2,\dots,m_N,n_N)$ where~$m_j=|\mathcal I_j|$ and~$n_j=|\mathcal O_j|$ (the size of the input and output spaces of the parties), a list of local operations $(f_1,f_2,\dots,f_N)$ where~$f_j:\mathcal I_j\rightarrow \mathcal O_j$, and a ``process function''~$w:\mathcal O\rightarrow\mathcal I$.
	We assume that the local operations and the ``process function'' are described as circuits over AND, OR, and NOT, as otherwise, complex computation could be hidden within these functions.
	If, under this choice of local operation and ``process function,'' party~$1$ receives some state~$i_1=1z\in\mathcal I_1$ for some~$z$, then we say~$\mathcal B$ {\em accepts\/}~$x$, otherwise,~$\mathcal B$ {\em rejects\/}~$x$.
	The algorithm~$\mathcal{B}$ {\em decides a language\/}~$L\subseteq\{0,1\}^*$ if~$\mathcal{B}$ accepts every~$x\in L$ and rejects every~$x\not\in L$.
	The class~$\cc{P_\text{LCCTC}}$ contains all languages decidable by some classical deterministic CTC algorithm. 
\end{definition}
The following inclusion relation follows immediately from the definitions and the Theorem above.
\begin{corollary}
	$\cc{P_\text{LCCTC}}\subseteq\cc{UP}\cap\cc{coUP}$.
\end{corollary}
\begin{proof}
	Assume a language~$L$ is in~$\cc{P_\text{LCCTC}}$.
	Then~$L$ is also in~$\cc{P_\text{NCCirc}}$ as we can construct a logically consistent circuit~$C$ which has the induced function~$c=w\circ f$.
\end{proof}

Finally, we discuss the respective search problems.
\begin{definition}[$\cc{FUP}$]
	A binary relation~$R$ is in~$\cc{FUP}$ (Function~$\cc{UP}$) if and only if~$R$ is polynomially decidable, polynomially balanced,
	and~$\forall x:|\{y\,|\,(x,y)\in R\}|\leq 1$.
\end{definition}
Informally, a problem is in~$\cc{FUP}$ if for every instance there exists {\em at most\/} one solution.
\begin{definition}[$\cc{F}(\cc{UP}\cap\cc{coUP})$]
	A pair~$(R_1,R_2)$ of relations is in~$\cc{F}(\cc{UP}\cap\cc{coUP})$ if and only if
	both relations are polynomially decidable, polynomially balanced, and for every instance~$x$
	\begin{align}
		&\left(  \exists!y:(x,y)\in R_1 \wedge\forall z: (x,z)\not\in R_2 \right)\oplus\\
		&\left(  \forall y:(x,y)\not\in R_1 \wedge \exists!z:(x,z)\in R_2 \right)
	\end{align}
	holds.
	The exclusive or~($\oplus$) asks for {\em either\/} yet {\em not both\/} expressions to be true.
\end{definition}
Note that the output of a search problem in~$\cc{F}(\cc{UP}\cap\cc{coUP})$ is some string~$w$ that satisfies either~$(x,w)\in R_1$ or (exclusively)~$(x,w)\in R_2$ but, as we formulated it, does not tell us in {\em which\/} relation the pair~$(x,y)$ appears.
However, since both relations are polynomially decidable, we can check in polynomial time whether~$y$ is a solution of~$R_1$ or~$R_2$.
This brings us to the following class, which is equal.
\begin{definition}[$\cc{TFUP}$]
	A binary relation~$R$ is in~$\cc{TFUP}$ (Totally~$\cc{FUP}$) if and only if~$R$ is polynomially decidable, polynomially balanced, and~$\forall x,\exists!y:(x,y)\in R$.
\end{definition}
\begin{theorem}
	\label{thm:tfup}
	$\cc{TFUP}=\cc{F}(\cc{UP}\cap\cc{coUP})$.
\end{theorem}
\begin{proof}
	Let~$R$ be a relation in~$\cc{TFUP}$ and~$R_1,R_2$ two relations such that for every~$x$:
	\begin{align}
		&\left(  \exists!y:(x,y)\in R_1 \wedge\forall z: (x,z)\not\in R_2 \right)\oplus\\
		&\left(  \forall y:(x,y)\not\in R_1 \wedge \exists!z:(x,z)\in R_2 \right)
		\,.
	\end{align}
	To show~$\cc{TFUP}\subseteq\cc{F}(\cc{UP}\cap\cc{coUP})$, set~$R_1=R$ and~$R_2=\emptyset$,
	and to show~$\cc{F}(\cc{UP}\cap\cc{coUP})\subseteq\cc{TFUP}$, set~$R=R_1\cup R_2$.
\end{proof}
A similar statement~$\cc{TFNP}=\cc{F(NP}\cap\cc{coNP})$ can also be made~\cite{Megiddo1991}.
The complexity class~$\cc{TFNP}$ is the class of all {\em total\/} relations that are polynomially decidable and polynomially balanced.

We now state and prove our last theorem.
\begin{theorem}
	$\cc{FP_\text{NCCirc}}=\cc{TFUP}$.
\end{theorem}
\begin{proof}
	We start with~$\cc{TFUP}\subseteq\cc{FP_\text{NCCirc}}$.
	A binary relation~$R$ in~$\cc{TFUP}$ is polynomially decidable and polynomially balanced.
	Therefore, there exists an algorithm~$\mathcal{D}$ that takes two inputs~$x,y$, runs in polynomial time in~$|x|$, and if~$(x,y)\in R $ then~$\mathcal{D}$ outputs ``yes,'' otherwise,~$\mathcal{D}$ outputs ``no.''
	Furthermore, for every instance~$x$ there exists a {\em unique\/}~$y$ with~$(x,y)\in R$.
	The deterministic NCCirc algorithm~$\mathcal{A}$, upon receiving~$x$, generates the circuit~$\mathcal{C}_x$ that acts as
	\begin{align}
		c_x:y\mapsto
		\begin{cases}
			y&(x,y)\in R\,,\\
			y'&\text{otherwise,}
		\end{cases}
	\end{align}
	where, if~$y=bz$ with~$b\in\{0,1\}$, then~$y'=(b\oplus 1)z$.
	Thus, for every~$x$ we have a circuit~$\mathcal{C}_x$ with the unique fixed point that equals the solution, {\it i.e.},~$c_x(y)=y\implies (x,y)\in R$.
	The converse inclusion relation~$\cc{FP_\text{NCCirc}}\subseteq\cc{TFUP}$ is shown as follows.
	Suppose we are given a relation~$R$ that is decidable by a deterministic NCCirc algorithm~$\mathcal{A}$.
	We now need to show that~$R$ is polynomially decidable, polynomially balanced, and that every~$x$ has a {\em unique\/} solution.
	Indeed,~$R$ is polynomially decidable and polynomially balanced because~$\mathcal{C}_x$ is generated in polynomial time, and~$\mathcal{C}_x$ upon input~$y$ is computed in polynomial time in~$|x|$.
	Furthermore,~$\mathcal{C}_x$ has a {\em unique\/} fixed point.
	The algorithm~$\mathcal{D}$ to decide~$R$ on input~$x$ returns the truth value of~$c_x(y)=y$.
\end{proof}

\section{Example:\ Integer factorization}
We give an example of an algorithm to factorize integers.
The NCCirc algorithm~$\mathcal{A}$ outputs, on input~$N\in\mathbb{Z}$,~a circuit~$\mathcal{C}_N$ with which~$N=p_1^{e_1}p_2^{e_2}\dots$ can be decomposed into its prime factors~$p_1,p_2,\dots$ along with its exponents~$e_1,e_2,\dots$.
We give a description of~$\mathcal{C}_N$ as an algorithm.
Clearly, this algorithm can be transformed into a circuit.
The following algorithm runs in a time polynomial in~$n=\lceil\log{N}\rceil$.
\begin{algorithm}[H]
	\caption{{\sf Factoring}~$N$}
	\label{alg:fact}
	\begin{algorithmic}[1]
		\Require $b\in\{0,1\},a_1,a_2,\dots,a_n,e_1,e_2,\dots,e_n\in K$
		\Ensure $b'\in\{0,1\},a_1,a_2,\dots,a_n,e_1,e_2,\dots,e_n\in K$
		\State $w \leftarrow \neg b,a_1,a_2,\dots,a_n,e_1,e_2,\dots,e_n$
		\For{$i=1 \text{ to } n-1$}
			\If{$(a_i < a_{i+1}) \vee (a_i\not=1 \wedge a_i=a_{i+1})$}
				\State\Return $w$
			\EndIf
		\EndFor
		\For{$i=1 \text{ to } n$}
			\If{$(a_i=1 \wedge e_i>1) \vee a_i \not\in \textsf{PRIME}\cup\{1\}$}
				\State\Return $w$
			\EndIf
		\EndFor
		\If{$a_1^{e_1}a_2^{e_2}\dots a_n^{e_n}\not=N$}
			\State\Return $w$
		\EndIf
		\State\Return $0,a_1,a_2,\dots,a_n,e_1,e_2,\dots,e_n$
	\end{algorithmic}
\end{algorithm}
Algorithm~\ref{alg:fact} takes as input~$1$ bit and~$2n$ numbers in~$K=\{1,2,\dots,N-1\}$, where every number is represented as an~$n$-bit string.
On line~$3$ we check whether the first~$n$ numbers are ordered.
On line~$8$ we check whether~$e_i$ is~$1$ whenever~\mbox{$a_i=1$}, and whether~$a_i$ is indeed prime (or~$1$).
A deterministic primality test can be performed in polynomial time as was recently shown~\cite{Agrawal2004}.
Finally, on line~$12$ we check whether the decomposition is correct.
If all tests pass, then the algorithm returns~$0,a_1,a_2,\dots,a_n,e_1,e_2,\dots,e_n$ where~$\prod_{i=1}^n a_i^{e_i}=N$, otherwise, the algorithm {\em flips\/} the first input bit.
This algorithm and, therefore, the circuit~$\mathcal{C}_N$, has the {\em unique\/} fixed point~$0,p_1,p_2,\dots,p_m,1^{n-m},e_1,e_2,\dots,e_m,1^{n-m}$, where~$p_1>p_2>\dots>p_m$ are primes and~$\prod_{i=1}^m p_i^{e_i}=N$.
Intuitively, one can understand this algorithm as ``killing the grandfather'' whenever a wrong factorization is given --- which resembles an instantiation of ``anthropic computing''~\cite{Aaronson2005} or ``quantum suicide''~\cite{Tegmark1998}.

\section{Conclusion and open questions}
The non-causal circuit model describes circuits where the assumption of a {\em global causal order\/} is replaced by the assumption of {\em logical consistency\/} ({\it i.e.}, no grandfather antinomy and no uniqueness ambiguity).
The problems that are solvable in polynomial time by such circuits form the complexity class~$\cc{P_\text{NCCirc}}$.
We show that this class equals~$\cc{UP}\cap\cc{coUP}$, where~$\cc{UP}$ consists of all problems in~$\cc{NP}$ which have an {\em unambiguous\/} accepting path.
Notable problems within~$\cc{UP}\cap\cc{coUP}$ are integer factorization and parity games.
Intuitively, the class~$\cc{P_\text{NCCirc}}$ contains all search problems that can be solved by determining the {\em unique\/} fixed point of a specific reformulation of the problem.
This gives a new interpretation of the class~$\cc{UP}\cap\cc{coUP}$.
The {\em uniqueness\/} requirement can be understood as arising from the assumption of no {\em overdetermination\/} (grandfather antinomy) and of no {\em underdetermination\/} (uniqueness ambiguity).
Similar complexity classes to~$\cc{FP_\text{NCCirc}}$ (the functional equivalent of~$\cc{P_\text{NCCirc}}$) are~$\cc{FIXP}$ and~$\cc{linear-FIXP}=\cc{PPAD}$~\cite{Etessami2010}.
Problems within these classes are fixed-point problems where {\em multiple\/} fixed points might exist, and in~$\cc{FIXP}$, the fixed points are allowed to be {\em irrational}.
Finding a Nash equilibrium for two parties is~$\cc{linear-FIXP-complete}$, and the same problem for three parties or more is~$\cc{FIXP-complete}$~\cite{Etessami2010}.
The class~$\cc{P_\text{NCCirc}}=\cc{UP}\cap\cc{coUP}$ is not believed to contain {\em complete\/} problems~\cite{Sipser}.

This result leads us to conclude that classical deterministic closed timelike curves, based on the framework for correlations without global causal order, {\em cannot\/} efficiently solve problems outside of~$\cc{UP}\cap\cc{coUP}$, {\it i.e.},~$\cc{P_\text{LCCTC}}\subseteq\cc{P_\text{NCCirc}}$.
The reason for this is that in the CTC model we require the composed map of the parties with the environment to have a unique fixed point for {\em any choice\/} of local operations of the parties.
This assumption was dropped when defining the non-causal circuit model.
However, note the subtlety that the framework for classical correlations without causal order (as opposed to the classical deterministic CTC model) could, then again, efficiently solve problems not solvable by classical deterministic CTCs.
The reason for this is that in the correlations framework, {\em fine-tuned\/} process matrices are allowed~\cite{Baumeler2016} which are inherently probabilistic --- here, we focused on {\em deterministic\/} CTCs instead.

When we compare this result to the computational power of the Deutsch CTC model, we note that the CTC model studied here is {\em dramatically\/} weaker.
This (possibly extreme) drop of computational power could be explained by the assumption of {\em linearity\/} which, in contrast to Deutsch's model, is present in the model studied here.
It is known that non-linearity can lead to astonishing results~\cite{Gisin1990,Polchinski1991,Abrams1998}.
Put differently, the absence of the grandfather antinomy allows to efficiently solve problems in~$\cc{PSPACE}$, yet, if we additionally ask for the absence of the uniqueness ambiguity, the computational power drops down to~$\cc{UP}\cap\cc{coUP}$.
In a similar spirit, the Deutsch version of CTCs restricted to {\em deterministic\/} fixed points gives a power of at most~$\cc{NP}\cap\cc{coNP}$~\cite{haltingp}.

One can put this result in the following perspective:
Previous results on closed timelike curves show that CTCs are not problematic from a general relativity theory point of view, from a logic point of view, and now we show their relative {\em innocence\/} from a computational point of view.

Some of the main open questions that remain are:
Does~$\cc{P_\text{LCCTC}}\supseteq\cc{P_\text{NCCirc}}$ hold or not, what are the probabilistic~($\cc{BPP_\text{NCCirc}}$, $\cc{BPP_\text{LCCTC}}$) and the quantum~($\cc{BQP_\text{NCCirc}}$, $\cc{BQP_\text{LCCTC}}$) versions of the complexity classes defined here, and how does~$\cc{BQP}$ relate to~$\cc{P_\text{NCCirc}}$?

\vspace{0.75cm}{\bf Data accessibility.}
This paper has no data.

\vspace{0.75cm}{\bf Competing interests.}
We have no competing interests.

\vspace{0.75cm}{\bf Authors' contributions.}
Both authors contributed equally on the results and on the paper.
Both authors gave final approval for publication.

\vspace{0.75cm}{\bf Acknowledgments.}
	We thank Adarsh Amirtham, Veronika Baumann, Gilles Brassard, Harry Buhrman, Paul Erker, Arne Hansen, and Alberto Montina for helpful discussions.
	We thank Claude Cr{\'e}peau for his kind invitation to the 2016 Bellairs Workshop, McGill Research Centre, Barbados, where the main ideas emerged.
	We thank four anonymous referees for helpful comments.

\vspace{0.75cm}{\bf Funding.}
	This work was supported by the Swiss National Science Foundation (SNF) and the National Centre of Competence in Research ``Quantum Science and Technology'' (QSIT).

\bibliography{refs}

\end{document}